\documentclass[12pt]{article}
\usepackage{hyperref} 
\usepackage{tikz,tkz-berge}
\usepackage{amsmath}
\usepackage{amsthm}
\usepackage{amssymb}
\usepackage{authblk}
\usepackage{float}

\usetikzlibrary{shapes,arrows,positioning,calc}

\newtheorem{theorem}{\bf Theorem}[section]
\newtheorem{corollary}[theorem]{\bf Corollary}
\newtheorem{lemma}[theorem]{\bf Lemma}

\newtheorem{definition}[theorem]{\bf Definition}

\title{Generalized Vertex Transitivity in Graphs}
\author{Kannan Balakrishnan\thanks{Email: mullayilkannan@gmail.com}}
\author{Divya Sindhu Lekha\thanks{Email: divi.lekha@gmail.com}}
\affil{Department of Computer Applications, Cochin University of Science and Technology, Kochi, IN-682022}
\author{Manoj Changat\thanks{Email:mchangat@gmail.com}}
\affil{Department of Futures Studies, University of Kerala, Thiruvananthapuram, IN-695581}
\author{Bijo S. Anand\thanks{Email: bijos\_anand@yahoo.com}}
\affil{Department of Mathematics, Sree Narayana College, Punalur, Kollam,IN-691305}
\author{Prasanth G. Narasimha-Shenoi\thanks{Email: prasanthgns@gmail.com}}
\affil{Department of Mathematics, Government College Chittur, Palakkad,IN-678104}
\begin{document}
	\date{}
	\maketitle
	\begin{abstract}
		In this paper, we introduce a generalization of vertex transitivity in graphs called generalized vertex transitivity. We put forward a new invariant called transitivity number of a graph. The value of this invariant in different classes of graphs is explored. Also, different results showing the importance of this concept is established.
		\flushleft{Keywords: Graph automorphism, Vertex transitivity,Graph products}
	\end{abstract}
	
	\section{Introduction}
	Many complex phenomena and systems are modeled using graphs where the entities are represented by vertices and their relations by corresponding edges. A graph in which all vertices are  equivalent  is known as a vertex-transitive graph~\cite{CCh95}. It is clearly evident that vertex-transitive graphs are also regular, but not vice versa. Studying the properties of a single vertex in such a graph gives insight into the properties of the complete vertex set. 
	
	The notions of vertex-transitive graphs and the theory of transitive permutations groups brought together two important notions in Mathematics namely Group Theory and Graph Theory. Each finite permutation group corresponds to vertex transitive graphs. Almost all vertex transitive graphs give rise to many transitive permutation groups of the automorphism group of a graph. In short, the theory of vertex-transitive graphs has developed in parallel with the theory of transitive permutation groups, see~\cite{praeger1997finite}. Also Praeger et.al, in~\cite{praeger1997finite}, explored the different ways in which these two theories have influenced each 
	other. 
	
	Many important classes of graphs are not vertex transitive, but can be partitioned into subsets in which all vertices in the same partition are equivalent. This gives us the ability to concentrate on a smaller subset of the vertices in order to gain insights into the properties of the whole graph. Also this will lead to efficient computation of various graph invariants,as well as graph properties. 
	
	By a permutation, we mean a one-to-one and onto function from a set onto itself. We say that two graphs $G$ and $H$ are isomorphic, if there exists a bijection $\phi : G\rightarrow H$ such that $\{u,v\}\in E(G)$ if and only if $\{\phi(u),\phi(v)\}\in E(H)$. If $G=H$, then these bijections are called permutations and the isomorphisms are known as automorphisms. In other words, an \textit{automorphism} of a graph $G$ is a permutation of the vertices, $V(G)$, of $G$ that preserves adjacency. For a  formal definition of automorphism and automorphism group of $G$ see~\cite{Cam05}.
	A relation $R$ on a non-empty set $A$ is a subset of $A\times A$. $R$ is said to be \textit{reflexive} if $(a,a)\in R$ for all $a\in A$.  We say that $R$ is \textit{symmetric} whenever $(a,b)\in R\implies (b,a)\in R$ and $R$ is said to be \textit{transitive} if $(a,b), (b,c)\in R\implies (a,c)\in R$. $R$ is said to be an equivalence relation if $R$ satisfies the reflexive, symmetric and transitive properties. Every equivalence relation on a set $A$ partitions the set$A$, and each partition is called an \textit{equivalence class}. For a detailed study refer~\cite{fraleigh2003first}.
	
	A group $(G, *)$ is a non-empty set $G$, closed under a binary operation $\ast$, which satisfies associative property. Therefore, there exists an element $e$ in $G$ such that for all $a\in G$ we have $a*e=a=e*a$, and for each element $a\in G$, there exists a unique element $a'\in G$ such that $a*a'=e=a'*a$.  For simplicity, we write $ab$ instead of $a*b$. The set of all automorphisms on $G$ denoted by $Aut(G)$ is a group under the binary operation, $\odot$, the composition of permutations.  
	
	A graph $G$ is said to be \emph{vertex-transitive} if every vertex in $G$ can be mapped to any other vertex by some automorphism.  That is, given any two vertices $v_1$ and $v_2$ of $G$, there is some automorphism $ \phi:V(G)\rightarrow V(G)$ such that $\phi(v_{1})=v_{2}$. In this case, we say that $Aut(G)$ acts transitively on $V(G)$~\cite{Cam05}. For an in-depth understanding of the properties of vertex-transitive graphs, see works by Sabidussi~\cite{Sab64}, Babai~\cite{Bab91} and Cameron~\cite{Cam05}. 
	
	Many authors have attempted the construction of vertex-transitive graphs with a given order $n$. The very first attempt in this direction was by Yap~\cite{Yap73}, in which he found all vertex-transitive graphs of order up to $11$. Another attempt was by McKay in 1979~\cite{McK79} and 1990~\cite{McR90}, to find out graphs of order up to $19$ and at most $26$, respectively. Vertex-transitive graphs of order $p$, where $p$ is a prime number, are graphs with a $p$-cycle in its automorphism groups~\cite{Tur67}. Similarly, graphs of order which is a product of two primes attained special attention in studies like ~\cite{Mas94, PrX93, ZhF10}.
	
	One of the important class of vertex-transitive graphs which are widely studied in the century are the Cayley graphs which are the graphs constructed using finite groups. Some examples of Cayley graphs are Complete graphs and their complements, complete multi-partite graph $K_{r(s)}$, and $k$-dimensional cubes $Q_k$. The prominence of the study of Cayley graphs owes to its immense applications in research of networks. In networks, the major agenda is to preserve connectivity even if some nodes or links are compromised. In such situations, vertex-transitive graphs, particularly Cayley graphs~\cite{Als05} play an important role. Cayley graphs are used extensively in modelling interconnection networks. See ~\cite{CCS85}, ~\cite{Hey97},~\cite{LJD93}. A very recent study~\cite{LiL18} shows that symmetric properties of a graph are very relevant in the behaviour of interconnection networks.  
	%
	%

	In this paper, we investigate on the concept of generalized vertex transitivity in graphs and some important graph products. All graphs considered in this paper are connected, simple and undirected. 
	
	We formally define the relevant graph operations here. 
	Let $G$ and $H$ be two graphs. Then their \emph{Cartesian product} $G \Box H$ is the graph with vertex set $V(G) \times V(H)$ where vertices $(g, h)$ and $(g', h')$ are adjacent if either\\
	\begin{itemize}
		\item $gg' \in E(G)$ and $h = h'$, or
		\item $g = g'$ and $hh' \in E(H)$.
	\end{itemize}
	The {\it strong product} of graphs $G$ and $H,$ denoted by
	$G\boxtimes H,$ has vertex set $V(G)\times V(H),$ where two distinct
	vertices $(x_1, y_1)$ and $(x_2, y_2)$ are adjacent if,
	\begin{itemize}
		\item $x_1 = x_2$ and $y_1y_2 \in E(H),$ or
		\item $y_1 = y_2$  and  $x_1x_2\in E(G),$ or
		\item $x_1x_2\in  E(G)$ and $y_1y_2\in E(H)$.\\
	\end{itemize}
	The \textit{corona product} $G\lozenge H$ of two graphs $G$ and $H$ (where $G$ has $p$ vertices) is defined as the graph obtained by taking one copy of $G$ and $p$~copies of $H$, and joining the $i^\text{th}$ vertex of $G$ to each and every vertex in the $i^\text{th}$ copy of $H$ by an edge.
	The \textit{join} of two graphs $G$ and $H$ is a graph formed from disjoint copies of $G$ and $H$ by connecting each vertex of $G$ to each vertex of $H$.
	\section{Generalized vertex transitivity in graphs}
	\begin{definition}\label{def:int}
		A vertex $u$ in graph $G$ is said to be \emph{interchangeable} with vertex $v$ in $G$ if there is an automorphism $\sigma$ of $G$ such that $\sigma(u)=v$. We denote this by $u\Omega v$. 
	\end{definition}
	\begin{theorem}
		Interchangeability $\Omega$ is an equivalence relation on $V$.
	\end{theorem}
	\begin{proof}
		In order to prove that a relation is an equivalence, we need to prove it is reflexive, symmetric and transitive.
		
		\begin{description}
			\item[$\Omega$ is reflexive] Since every vertex $u$ in $G$ is mapped on to itself by the identity automorphism.
		\end{description} 
		\begin{description}
			\item[$\Omega$ is symmetric] If $u \Omega v$ then there is an automorphism $\sigma$ such that $\sigma(u) = v$. Now, the automorphism $\sigma^{-1}$ has the property that $\sigma^{-1}(v) = u$.
		\end{description}
		\begin{description}
			\item[$\Omega$ is transitive] If $u \Omega v$ and $v \Omega w$ then there exists automorphisms $\sigma_1$, $\sigma_2$ such that $\sigma_1(u) = v$ and $\sigma_2(v) = w$. Now, the composite $\sigma_2 \circ \sigma_1$ maps $u$ to $w$.
		\end{description}
		
		Since $\Omega$ is reflexive, symmetric and transitive; it is an equivalence relation.
	\end{proof}
	
	Since $\Omega$ is an equivalence relation, we can say that two vertices $u$ and $v$ are interchangeable if $u\Omega v$. Also, $\Omega$ partitions the vertex set $V(G)$ into a set of vertex classes; any two vertices in the same class are interchangeable. We denote this set by $\frac{V}{Aut(G)} = \{V_1, V_2, \ldots, V_r\}$. Each $V_i$ is called a transitivity partition or simply a partition. 
	
	\begin{definition}\label{def:num}
		\item[Transitivity number of a graph $G$,  $r_G$], is the number of transitivity partitions in $G$. $G$ is said to be $r$-transitive if its transitivity number is $r$.
	\end{definition}
	
	\begin{definition}\label{def:rep}
		\item[Transitive representation of a graph $G$,  $R_G$,] is the set $\{v_1, v_2, \ldots, v_r\}$ in which each $v_i \in V_i$ where $V_i$ is a transitivity partition in $G$. 
	\end{definition}
	
	We denote the partition in which a vertex $v$ belongs as $R(v)$. Clearly, $R(v) = Aut(G)(v)$. i.e. $R(v)$ is the same as the set of all $\sigma(v)$ where $\sigma$ is some automorphism of $G$. Also, every vertex in $R(v)$ has the same vertex properties like degree, eccentricity and vertex centralities. 
	\subsection{ Functions on vertex set}
	From above discussions, it is evident that the transitivity number of a graph is greater than or equal to the number of distinct elements in the degree sequence of the graph. Also, members of the same partition have same values for any function on $V$ which is preserved under graph automorphism. This fact has vast implications as we show in the next few theorems.
	See the obvious lemma, Lemma~\ref{lem:aut}.
	\begin{lemma}\label{lem:aut}
		Let $G$ be a graph. Then every automorphism of $G$ maps the sets $V_i$ onto themselves.
	\end{lemma}
	\begin{definition}
		Let $f$ be a function on graph $G$ with vertex set $V$; $f: V^* \rightarrow 2^V - \phi$. $M$ is said to be preserved under automorphisms if for every automorphism $\sigma$ , $\sigma(f(V)) = f(\sigma(V))$.
	\end{definition}
	Now we prove a more powerful lemma.
	
	\begin{lemma}\label{lem:fn}
		Let $f$ be a function on $G$, $f: V^* \rightarrow 2^V - \phi$, which is preserved under automorphisms. Let $V_1, V_2, \ldots, V_r$ be the partitions in $G$. Then $f(V_i)$ is always the union of some $V_j$s.
	\end{lemma}
	\begin{proof}
		Let $V_i = \{v_{i_1}, v_{i_2}, \ldots, v_{i_q}\}$ be a partition. Then $f(V_i) = \bigcup\limits_{t= 1}^q f(\{v_{i_t}\})$. We need to prove that $f(V_i) = \bigcup\limits_{j= 1}^pV_j$; $p \le r$.
		
		Let $f(\{v_{i_k}\}) = \{u_1, u_2, \ldots, u_k\} = U_k$ where $k \in (1,t)$. Consider vertex $u_l \in U_k$, $l\in (1,k)$. 
		
		Let $V_m = R(u_l)$, the partition containing $u_l$ as an element. By definition of interchangeability~\ref{def:int}, we can say that :
		$$\forall x \ (x \in V_m) \implies (x \in U_k)$$
		Thus, $f(\{v_{i_k}\}) = U_k = \bigcup\limits_{j= 1}^pV_j$; $p\le r$. 
		
		The same holds true for all $v_{i_k} (k \in (1,q))$. Therefore, without loss of generality, we can say that $f(V_i) = \bigcup\limits_{j= 1}^pV_j$; $p \le r$.
	\end{proof}
	\begin{theorem}\label{thm:fn}
		Let $G$ be an $r$-transitive graph with vertex set $V$ and partitions $V_1, V_2, \ldots, V_r$. Let $f$ be a function defined on $G$, $F: V^* \rightarrow 2^V - \phi$, which is preserved under automorphism. Then $f(V)$ is always the union of one or more $V_i$. 
	\end{theorem}
	\begin{proof}
		Since $V_1, V_2, \ldots, V_r$ are partitions of $V(G)$, we can write $f(V)$ as $f(V_1) \bigcup f(V_2) \bigcup \ldots f(V_r)$. Now, by lemma~\ref{lem:fn}, we know that $f(V)$ is a union of one or more $V_i$. Thus proves the theorem.
	\end{proof}
	
	Theorem~\ref{thm:fn} holds for functions including consensus functions like center,betweenness center and median of a graph. 
	
	\begin{lemma}\label{lem:fnR}
		Let $G$ be an $r$-transitive graph with vertex set $V$. Let $f$ be a function defined on $V$, $f: V \rightarrow \mathbb{R}$, which is preserved under automorphism. Then the maximum number of distinct values in $f(V)$ is $r$.
	\end{lemma}
	\begin{proof}
		Let $V_1, V_2, \ldots, V_r$ be the partitions in $G$. 
		$$f(V) = f(V_1) \bigcup f(V_2) \bigcup \ldots f(V_r) \textnormal{ (By lemma~\ref{lem:fn}}).$$
		Since $f$ is preserved under automorphism, we can say that $\sigma(f(V_i)) = f(V_i)$ where $\sigma$ is an automorphism. Therefore there is only one unique value in $f(V_i)$. Thus the maximum number of distinct values for $f(V)$ is $r$.
	\end{proof}
	As an example, let us consider the \emph{total distance $D$} of a vertex $v$ in an $r$-transitive graph $G$ with vertex set $V$. 
	$$D(v) = \sum\limits_{u \in V}d(u,v).$$
	Let $v_1$ and $v_2$ be two vertices in a partition $V_i$. Since $v_1 \Omega v_2$, there exists an automorphism $\sigma$ such that $\sigma(v_1) = v_2$. Therefore $D(v_1) = D(v_2)$. Hence, for each $v_i \in V_i$, the total distance $D(v_i)$ is the same value. Thus, there are $r$ distinct values of $D$ in a $r$-transitive graph.
	
	The following theorem is a direct implication of the lemma~\ref{lem:fnR}.
	\begin{theorem}\label{thm:fnR}
		Let $G$ be an $r$-transitive graph with vertex set $V$. Let $f$ be a function defined on $V$, $f: V \rightarrow \mathbb{R}$, which is preserved under automorphism. Then 
		$ |\{f(v_i)\}|\le r,\ v_i \in V$
	\end{theorem}
	\begin{proof}
		Lemma~\ref{lem:fnR} implies that there are maximum $r$ distinct values for $f(V)$. Therefore, for every $v_i \in V,\  |\{f(v_i)\}|\le r$
	\end{proof}
	
	Next, we focus on the vertex properties which are invariant under automorphism.
	\begin{definition}
		\item[A vertex property] $P: V \times 2^V \rightarrow \mathbb{R}$ is invariant under automorphisms if for every automorphism $\sigma$ , $P(\sigma(V)) = P(V)$.
	\end{definition}
	Examples are vertex degree $\delta(v)$, vertex eccentricity $e(V) = \max\limits_{u \in V}d(u,v)$ and total distance $D(v)=\sum\limits{u \in V}d(u,v)$.
	\begin{lemma}\label{lem:Prop}
		Let $G$ be an $r$-transitive graph with partitions $V_1, V_2, \ldots, V_r$. Let $P$ be a vertex property of $G$ invariant under $Aut(G)$. Then 
		$$\{x|P(x)=y\} =\bigcup V_i \ or\  \phi$$
	\end{lemma}
	\begin{proof}
		$P$ is a vertex property of $G$. Therefore $P$ is invariant under graph automorphisms. Let $\sigma$ be an automorphism. Then, we can state that $P(\sigma(x)) = P(x)$. So, for any two vertices $u,v \in V_i (i=1, \ldots, r)$, $P(u) = P(v) = y.$ Therefore, $P(V_i) = y$. By lemma~\ref{lem:fnR}, we know that $P(V) \le r$. Let $\{x\} = \bigcup\limits_{i = 1}^k V_i$, $k \le r$. Then $y=P(x)$ is a set where $|y| \le r$. Thus, $\{x|P(x)=y\}$ is either a union of many partitions, $\bigcup V_i$, or an empty set.
	\end{proof}
	Note that this lemma gives rise to the following theorem.
	\begin{theorem}\label{thm:Prop}
		Let $G$ be an $r$-transitive graph and $P$ be a vertex property of $G$ invariant under $Aut(G)$. Then there are at most $r$ distinct values for $P$ in $G$.
	\end{theorem}
	\begin{proof}
		Lemma~\ref{lem:Prop} implies that if $P(x)=y$, then $\{x\}$ is either a union of partitions in $G$ or an empty set. Therefore, we can say that there is at least a unique value for $P(x)$ if x is a union of partitions $V_i$.  Since there are $r$ partitions in $G$, there can be a maximum of $r$ distinct values for $P$ in $G$.
	\end{proof}
	The power of generalized vertex transitivity can be utilized in inferring the behavior of graphs under expansions and contractions. In the next section, we consider the generalized vertex transitivity in Cartesian products. 
	\section{Transitivity Number and Graph Operations}
	First, we observe the Cartesian product of paths. 
	\begin{lemma}\label{lem:Pmn}
		Let $P_m$ be a path of size $m$ and $P_n$ be a path of size $n$ such that $m \neq n$.  Let $r_m$ and $r_n$ be their respective transitivity numbers. Then their Cartesian product is $r_mr_n$ transitive.
	\end{lemma}
	\begin{proof}
		Let $G_1 = P_m$ and $G_2 = P_n$ with partitions $\Pi_1 = \{U_1, U_2, \ldots, U_{r1}\}$ and $\Pi_2 = \{V_1, V_2, \ldots, V_{r2}\}$ respectively. Let $G = G_1 \Box G_2$.  Let $(a_1,a_2)$ and $(b_1,b_2)$ be two vertices in $G$. Then $(a_1, a_2)$ and $(b_1, b_2)$ are interchangeable if  either $a_1 = b_1$ and  $a_2, b_2 \in V_i$  where $V_i$ is a partition in $\Pi_2$, 
		or  $a_2 = b_2$ and  $a_1, b_1 \in U_j$  where $U_j$ is a partition in $\Pi_1$. Following this we can get the partition for $G$ as  
		\begin{align*}\Pi &= \{U_1 \times V_1, U_2 \times V_1,  \ldots, U_{r1} \times V_1, \\&U_1 \times V_2, U_2 \times V_2,  \ldots, U_{r1} \times V_2, \\& \ldots, \\& U_1 \times V_{r2}, U_2 \times V_{r2},  \ldots, U_{r1} \times V_{r2}  \}
		\end{align*}
		There are ($r_1 \times r_2$) unique elements in $\Pi$. Therefore the transitivity number of $G$ is $r_1 \times r_2$.
	\end{proof}
	
	\begin{lemma}\label{lem:Pn}
		Let $P$ be a path of size $n$ with transitivity number $r$. Then the transitivity number of the Cartesian product $P \Box P$ is $\frac{r \times (r+1)}{2}$.
	\end{lemma}
	\begin{proof}
		Let $\Pi_P = \{U_1, U_2, \ldots, U_{r}\}$ be the partitions of $P$.   Then by  we get the partitions for $G=P \Box P$ as  
		\begin{align*}\Pi_G &= \{U_1 \times U_1, U_2 \times U_1,  \ldots, U_{r} \times U_1, \\&U_1 \times U_2, U_2 \times U_2,  \ldots, U_{r} \times U_2, \\& \ldots, \\& U_1 \times U_{r}, U_2 \times U_{r},  \ldots, U_{r} \times U_{r}  \}
		\end{align*}
		among which $U_i \times U_j$ is same as $U_j \times U_i$ for all $i, j = 1,2, \ldots, r$.  Thus the number of unique elements in $\Pi_G$ is reduced to $\frac{r \times (r+1)}{2}$. Therefore the transitivity number of $G$ is $\frac{r \times (r+1)}{2}$.
	\end{proof}
	Next theorem is for Cartesian and strong product of Graphs
	\begin{theorem}[Cartesian and strong product]\label{cartesian}
		Let $G$ and $H$ be two non isomorphic graphs with transitivity numbers $m$ and $n$ respectively. Then the transitivity number of $G* H$ is $m\times n$, where $*\in \{\Box,\boxtimes\}$.
	\end{theorem}	
	\begin{proof}
		First let $*\in \{\Box,\boxtimes\}$. Let $G/\text{Aut}(G)=\{U_1, U_2,\ldots,U_m\}=\Pi_1 $ and $H/\text{Aut}(H)=\{V_1, V_2,\ldots,V_n\}=\Pi_2$. To prove $G* H/\text{Aut}(G* H)=\Pi_1\times \Pi_2$. Let $(x,y), (x',y')\in U_i\times V_j$. We have to prove there exists an automorphism $\tau$ on $G* H$ such that $\tau((x,y))= (x',y')$ and there does not exist an automorphism on $G* H$ with image of any vertex of $U_i\times V_j$ is in $V(G* H)\setminus (U_i\times V_j)$. Since $x,x'\in U_i$ and $y,y'\in V_j$, then there exist automorphisms $\phi :G\rightarrow G$ and $\psi :H\rightarrow H$ such that $\phi(x)=x'$ and $\psi(y)=y'$. Define $\eta : G* H\rightarrow G* H$ as $\eta((x,y))=(\phi(x),\psi(y))$. We can see that $\tau$ is a bijection. Also these $\tau$ partitions $G* H$ into $\Pi_1\times \Pi_2$. It is enough to prove $\tau$ is a graph isomorphism. If $(x,y), (x',y')$ are adjacent, then either $x=x', \{y,y'\}\in E(H)$ or $\{x,x'\}\in E(G), y=y'$. Let us assume $x=x'$. Then we can choose $\phi=e$, the identity mapping from $G$ to $G$. Since $\{y,y'\}$ is an edge, and since $\psi$ is an isomorphism, we can see that $\{\psi(y),\psi(y')\}$ is an edge, that is $(x,\psi(y)),(x', \psi(y'))$ is an edge. Hence $\tau $ preserves the isomorphism. Clearly these kinds of $\tau$ does not map a vertex from $U_i\times V_j$ to any vertex of $V(G* H)\setminus (U_i\times V_j)$.

		Suppose there is an automorphism $\tau: G* H\rightarrow G* H$ such that $\tau (x,y)=(x',y')$, where $(x,y)\in U_i\times V_j$ but $(x',y')\in U_p\times V_q$ where either $i\not=p$ or $j\not=q$ or $i\neq p,j\neq q$.  Without loss of generality, we can assume that $i\not=p$.  Let us define $\phi_1:G\rightarrow G$ as follows.  For a fixed $y\in H$, define $\phi_1(x)=p_G(\tau(x,y))=p_G(x',y')=x'$, where $p_G$ is the usual projection.  Since $y$ is fixed, we can see that $\phi_1$ is one -one and onto.  Also we can see that $\phi_1$ is edge preserving.  That means $\phi_1$ is an automorphism, mapping $U_i$ to $U_p$, which is a contradiction that $U_i$ and $U_p$ are different partitions in $G/\text{Aut}(G)$.

		But, in this case, the projection of $\tau $ on the factor $G$, is an automorphism which maps some element from $U_i$ to some element in $U_p$. But this leads to a contradiction that $U_i$ and $U_p$ are different partitions of $G/\text{Aut}(G)$.
		
		So we can conclude that $G* H/\text{Aut}(G* H)=\Pi_1\times \Pi_2$. Hence the transitive number of $G* H$ is $m\times n$.
	\end{proof}
	We have an immediate corollary for the theorem~\ref{cartesian}
	\begin{corollary}
		If  $G$ and $H$ be two isomorphic graphs with transitivity numbers $m$. Then the transitivity number of $G* H$ is $\frac{n(n+1)}{2}$, where $*\in \{\Box,\boxtimes\}$.
		
	\end{corollary}
	\begin{theorem}
		Let $G_1$ and $G_2$ be two path graphs of size $m$ and $n$ respectively. Let $r_m$ and $r_n$ be their respective transitivity numbers. Then the transitivity number of their Cartesian product $G = G_1 \Box G_2$ is $r \le r_m \times r_n$.
	\end{theorem}
	\begin{proof}
		Proof follows from lemma ~\ref{lem:Pmn} and lemma~\ref{lem:Pn}.
	\end{proof}
	
	
	\section{Concluding Remarks}
	In this paper, we formalize the notion of generalized vertex transitivity in graphs. We also introduce the graph invariant known as transitivity number $r$. Knowing the transitivity number of a graph can aid in efficient graph computations. Once we know that a graph $G$ is $r$ transitive with partitions $(V_1, \ldots, V_r)$ then we can devise fast algorithms for computing various properties of $G$. This is because, now we need not evaluate all the vertices in $G$, but only $r$ vertices representing each partition.
	\section*{Acknowledgements}
	This work was supported by the National Post Doctoral Fellowship (N-PDF) No. PDF/2016/002872 from Science and Engineering Research Board (SERB), Department of Science and Technology (DST), Government of India.
	

\begin{thebibliography}{5}
		\bibitem{Als05}\emph{Alspach, B.}: Cayley graphs, Chapter in \emph{Topics in Algebraic Graph Theory} Edited by Beineke, L. W. and Wilson, R., Cambridge Univ. Press (2005).
		\bibitem{Bab91}\emph{Babai, L.}: Vertex-transitive graphs and vertex-transitive maps, J. Graph Theory 15, pages. 587 - 627 (1991).
		\bibitem{CCh95}\emph{Chiang, W., Cheng, R.}: The $(n, k)$-Star Graph: A Generalized Star Graph., Information Processing Letters 56(5), pp. 259 -- 264 (1995)
		\bibitem{CCS85}\emph{Carlsson, G. E., Cruthirds, J. E, Sexton, H. B, Wright, C. O.}: Interconnection networks based on a generalization of cube-connected cycles. IEEE Trans. Comp. C., 34, pp. 769-772 (1985).
		\bibitem{Cam05}\emph{Cameron, P. J.}: Automorphisms of graphs, Chapter in \emph{Topics in Algebraic Graph Theory} Edited by Beineke, L. W. and Wilson, R., Cambridge Univ. Press (2005).
		\bibitem{fraleigh2003first}\emph{Fraleigh, J. B};A first course in abstract algebra, Pearson Education India (2003).
		\bibitem{Hey97}\emph{Heydemann, M. C.}: Cayley graphs and interconnection networks. In Graph symmetry: algebraic methods and applications (Editors: Hahn and Sabidussi), pages 167–226. Kluwer Academic Publishers, Dordrecht, (1997).
		\bibitem{LiL18}\emph{Li, J.J., Ling, B.}: Symmetric graphs and interconnection networks, Future Generation Computer Systems, 83, pages. 461 -467 (2018).
		\bibitem{LJD93}\emph{Lakshmivarahan, S., Jho, J-S., Dhall, S. K.}: Symmetry in interconnection networks based on Cayley graphs of permutation groups: A survey. Parallel Computing, 19:361–407, (1993).
		\bibitem{Mas94}\emph{Maru\v{s}i\v{c}, D., Scapellato, R.}: Classifying vertex-transitive graphs whose order is a product of two primes, Combinatorica 14, pages. 187–201 (1994).
		\bibitem{McK79}\emph{McKay, B. D.}: Transitive graphs with fewer than twenty vertices, Math. Comp. 33, pages. 1101-1121 (1979). 
		\bibitem{McR90}\emph{McKay, B. D., Royle, G. F.}: The Transitive Graphs with at Most 26 Vertices, ARS COMBINATORIA 30, pages. 161-176 (1990). 
		\bibitem{PrX93}\emph{Praeger, C. E., Xu, M. Y.}: Vertex-Primitive Graphs of Order a Product of Two Distinct Primes, Journal of Combinatorial Theory, Series B, Volume 59, Issue 2, pages. 245-266, (1993).
		\bibitem{praeger1997finite}\emph{Praeger, C. E. and Li, C. H. and Niemeyer, A. C}: Finite transitive permutation groups and finite vertex-transitive graphs, Graph symmetry, pp. 277--318, Springer (1997).
		\bibitem{Sab64}\emph{Sabidussi, G.}: Vertex-transitive Graphs,  Monatsh. Math. 68, pages: 426 - 438 (1964). 
		\bibitem{Tur67}\emph{Turner, J.}: Point-symmetric graphs with a prime number of points, J Combin. Theory 3, pages. 136–145 (1967).
		\bibitem{Yap73}\emph{Yap, H. P.}: Point symmetric graphs with $p \le 13$ points, Nanta Math. 6, pages. 8 - 20 (1973). 
		\bibitem{ZhF10}\emph{Zhou, J. X., Feng, Y. Q.}: Cubic vertex-transitive graphs of order $2pq$, Journal of Graph Theory, 65(4), pages. 285-302 (2010).
	\end{thebibliography}
\end{document}